\pdfoutput=1

\RequirePackage[l2tabu,orthodox]{nag}

\documentclass[11pt,letterpaper,reqno,oneside]{article}

\DeclareTextFontCommand{\emph}{\slshape}

\makeatletter
\renewcommand{\paragraph}{%
	\@startsection{paragraph}{4}%
	{\z@}{1.75ex \@plus 1ex \@minus .2ex}{-0.7em}%
	{\normalfont\normalsize\bfseries}%
}
\makeatother

\usepackage[T1]{fontenc}
\usepackage{lmodern}
\usepackage[utf8]{inputenc}

\usepackage[american,german,british]{babel}

\usepackage[activate={true,nocompatibility}]{microtype}

\usepackage{csquotes}

\usepackage[backend=biber,autolang=other,style=apa,maxcitenames=5,sortcites=false]{biblatex}
\DeclareLanguageMapping{british}{british-apa}
\DeclareLanguageMapping{american}{american-apa}

\usepackage{amsmath}
\usepackage{amssymb}
\usepackage{amsfonts}
\usepackage{amsthm}
\usepackage{mathtools}
\usepackage{stmaryrd}

\let\originalleft\left
\let\originalright\right
\renewcommand{\left}{\mathopen{}\mathclose\bgroup\originalleft}
\renewcommand{\right}{\aftergroup\egroup\originalright}

\usepackage{graphicx}
\usepackage{tikz,pgfplots}
\pgfplotsset{compat=1.10}
\usetikzlibrary{shapes.multipart}
\usetikzlibrary{decorations.markings}

\usepackage[title]{appendix}

\usepackage{datetime}
\newdateformat{datestyle}{ {\THEDAY} \monthname[\THEMONTH] \THEYEAR }

\usepackage{enumerate}
\usepackage{enumitem}
\setlist[enumerate,1]{label=(\arabic*)}
\setlist[itemize,1]{label=--}
\setlist[itemize,2]{label=--}
\setlist[itemize,3]{label=--}
\setlist[itemize,4]{label=--}

\usepackage{xcolor}
\usepackage{verbatim}
\usepackage{gensymb}
\usepackage{rotating}
\usepackage{subcaption}
\usepackage{floatpag}
\floatpagestyle{plain}
\usepackage{marvosym}

\usepackage{hyphenat}
\theoremstyle{definition}

\newtheorem{definition}{Definition}
\newtheorem*{claim}{Claim}

\newtheoremstyle{named}
	{\topsep}					
	{\topsep}					
	{}							
	{0pt}						
	{\bfseries}					
	{}							
	{5pt plus 1pt minus 1pt}	
	{\thmnote{#3}}				
\theoremstyle{named}
\newtheorem{namedthm}{}

\renewcommand{\qedsymbol}{$\blacksquare$}

\usepackage{xpatch}
\xpatchcmd{\proof}{\itshape}{\proofheadfont}{}{}
\newcommand{\proofheadfont}{\slshape}

\usepackage{hyperref}
\hypersetup{pdfborder=0 0 0}

\usepackage[nameinlink]{cleveref}
\crefname{page}{p.}{pp.}
\crefname{equation}{equation}{equations}
\crefname{section}{section}{sections}
\crefname{subsection}{section}{sections}
\crefname{subsubsection}{section}{sections}
\crefname{appsec}{appendix}{appendices}
\crefname{supplsec}{supplemental appendix}{supplemental appendices}
\crefname{footnote}{footnote}{footnotes}
\crefname{figure}{figure}{figures}
\crefname{table}{table}{tables}
\crefname{theorem}{theorem}{theorems}
\crefname{proposition}{proposition}{propositions}
\crefname{lemma}{lemma}{lemmata}
\crefname{corollary}{corollary}{corollaries}
\crefname{remark}{remark}{remarks}
\crefname{observation}{observation}{observations}
\crefname{example}{example}{examples}
\crefname{fact}{fact}{facts}
\crefname{definition}{definition}{definitions}
\crefname{assumption}{assumption}{assumptions}
\crefname{exercise}{exercise}{exercises}
\crefname{notation}{notation}{notation}
\crefname{claim}{claim}{claims}
\crefname{conjecture}{conjecture}{conjectures}

\newcommand{\eps}{\varepsilon}

\newcommand{\dd}{\mathrm{d}}

\newcommand{\DD}{\mathrm{D}}

\newcommand{\E}{\mathbf{E}}

\newcommand{\R}{\mathbf{R}}

\newcommand{\N}{\mathbf{N}}

\newcommand{\1}{\boldsymbol{1}}

\newcommand{\join}{\vee}

\newcommand{\meet}{\wedge}

\newcommand{\Union}{\bigcup}

\DeclarePairedDelimiter\abs{\lvert}{\rvert}

\DeclarePairedDelimiter\floor{\lfloor}{\rfloor}

\newcommand*{\xslant}[2][76]{%
	\begingroup
	\sbox0{#2}%
	\pgfmathsetlengthmacro\wdslant{\the\wd0 + cos(#1)*\the\wd0}%
	\leavevmode
	\hbox to \wdslant{\hss
		\tikz[
			baseline=(X.base),
			inner sep=0pt,
			transform canvas={xslant=cos(#1)},
		] \node (X) {\usebox0};%
		\hss
		\vrule width 0pt height\ht0 depth\dp0 %
	}%
	\endgroup
}
\makeatletter
\newcommand*{\xslantmath}{}
\def\xslantmath#1#{%
	\@xslantmath{#1}%
}
\newcommand*{\@xslantmath}[2]{%
	\ensuremath{%
		\mathpalette{\@@xslantmath{#1}}{#2}%
	}%
}
\newcommand*{\@@xslantmath}[3]{%
	\xslant#1{$#2#3\m@th$}%
}
\makeatother

\makeatletter
\def\namedlabel#1#2{\begingroup
	#2%
	\def\@currentlabel{#2}%
	\phantomsection\label{#1}\endgroup
}
\makeatother

\makeatletter
\let\save@mathaccent\mathaccent
\newcommand*\if@single[3]{%
	\setbox0\hbox{${\mathaccent"0362{#1}}^H$}%
	\setbox2\hbox{${\mathaccent"0362{\kern0pt#1}}^H$}%
	\ifdim\ht0=\ht2 #3\else #2\fi
	}
\newcommand*\rel@kern[1]{\kern#1\dimexpr\macc@kerna}
\newcommand*\widebar[1]{\@ifnextchar^{{\wide@bar{#1}{0}}}{\wide@bar{#1}{1}}}
\newcommand*\wide@bar[2]{\if@single{#1}{\wide@bar@{#1}{#2}{1}}{\wide@bar@{#1}{#2}{2}}}
\newcommand*\wide@bar@[3]{%
	\begingroup
	\def\mathaccent##1##2{%
	  \let\mathaccent\save@mathaccent
	  \if#32 \let\macc@nucleus\first@char \fi
	  \setbox\z@\hbox{$\macc@style{\macc@nucleus}_{}$}%
	  \setbox\tw@\hbox{$\macc@style{\macc@nucleus}{}_{}$}%
	  \dimen@\wd\tw@
	  \advance\dimen@-\wd\z@
	  \divide\dimen@ 3
	  \@tempdima\wd\tw@
	  \advance\@tempdima-\scriptspace
	  \divide\@tempdima 10
	  \advance\dimen@-\@tempdima
	  \ifdim\dimen@>\z@ \dimen@0pt\fi
	  \rel@kern{0.6}\kern-\dimen@
	  \if#31
	    \overline{\rel@kern{-0.6}\kern\dimen@\macc@nucleus\rel@kern{0.4}\kern\dimen@}%
	    \advance\dimen@0.4\dimexpr\macc@kerna
	    \let\final@kern#2%
	    \ifdim\dimen@<\z@ \let\final@kern1\fi
	    \if\final@kern1 \kern-\dimen@\fi
	  \else
	    \overline{\rel@kern{-0.6}\kern\dimen@#1}%
	  \fi
	}%
	\macc@depth\@ne
	\let\math@bgroup\@empty \let\math@egroup\macc@set@skewchar
	\mathsurround\z@ \frozen@everymath{\mathgroup\macc@group\relax}%
	\macc@set@skewchar\relax
	\let\mathaccentV\macc@nested@a
	\if#31
	  \macc@nested@a\relax111{#1}%
	\else
	  \def\gobble@till@marker##1\endmarker{}%
	  \futurelet\first@char\gobble@till@marker#1\endmarker
	  \ifcat\noexpand\first@char A\else
	    \def\first@char{}%
	  \fi
	  \macc@nested@a\relax111{\first@char}%
	\fi
	\endgroup
}
\makeatother

\makeatletter
\def\@seccntformat#1{\@ifundefined{#1@cntformat}%
	{\csname the#1\endcsname\quad}
	{\csname #1@cntformat\endcsname}}
\newcommand\section@cntformat{}
\makeatother

\usepackage{xr}
\title{\scshape Screening for breakthroughs:\\
Omitted proofs}

\author{%
Gregorio Curello \\
University of Mannheim
\and
Ludvig Sinander \\
University of Oxford}

\date{13 March 2024}

\makeatletter
	\AtBeginDocument{ \hypersetup{
		pdftitle = {Screening for breakthroughs: Omitted proofs},
		pdfauthor = {Gregorio Curello and Ludvig Sinander}
		} }
\makeatother

\begin{document}

\maketitle


\begin{abstract}
	This document contains all proofs omitted
	from our working paper `Screening for breakthroughs'
	\parencite{sfb};
	specifically, the March 2025 version of the paper ({\color{blue}\href{http://www.arxiv.org/abs/2011.10090v9}{version 9 on arXiv}}).
\end{abstract}

\section*{\Cref*{sec:pf:model} (p. \pageref*{sec:pf:model})}
\label{omit:pf:model}
\addcontentsline{toc}{section}{Section \ref*{sec:pf:model} (p. \pageref*{sec:pf:model})}

The two utility possibility frontiers may be written as
\begin{align*}
	F^0(u)
	&= u - \lambda \phi^{-1}(u) 
	\\
	\text{and} \quad
	F^1(u)
	&= u + \lambda \max_{L \geq 0} \left\{
	w L - \phi^{-1} \left( u + \kappa(L) \right)
	\right\} .
\end{align*}

\begin{proof}[Proof of \Cref*{lemma:UI_assns}]
	Clearly $F^0,F^1$ are well-defined and continuous on $\left[ 0, \infty \right)$.
	It remains to show that they are strictly concave with peaks $u^0,u^1$ satisfying $u^1 < u^0$, and that the gap $F^1 - F^0$ is strictly decreasing.

	\begin{namedthm}[Claim.]
		\label{claim:Lstar_interior}
		For each $u>0$, the maximisation problem in the expression for $F^1(u)$
		has a unique solution $L^\star(u) \in (0,\infty)$.
		Furthermore, $\lim_{u \to \infty} L^\star(u) = 0$.
	\end{namedthm}

	\begin{proof}[Proof]%
		\renewcommand{\qedsymbol}{$\square$}
		Fix $u>0$, and write $f : \R_+ \to \R$ for the objective. Its (right-hand) derivative is
		\begin{equation*}
			f'(L)
			= w - \frac{ \kappa'(L) }
			{ \phi'\left( \phi^{-1} \left( u + \kappa(L) \right) \right) } .
		\end{equation*}
		Clearly $f'$ is strictly decreasing, so there is at most one maximiser.
		We have $f'(0) = w > 0$, whereas $f'(L) < 0$ for any large enough $L>0$ since
		\begin{equation*}
			\lim_{L \to \infty} f'(L)
			\leq w - \lim_{L \to \infty} \frac{ \kappa'(L) }
			{ \phi'\left( \phi^{-1} \left( \kappa(L) \right) \right) }
			= - \infty ,
		\end{equation*}
		where the equality holds since
		the numerator is bounded away from zero as $L \to \infty$
		while the denominator vanishes.
		Thus the unique maximiser $L^\star(u)$ of $f$ is interior, and therefore satisfies the first-order condition $f'\left( L^\star(u) \right) = 0$.
		As $u \to \infty$ in the first-order condition, the denominator in the fraction vanishes since $\lim_{C \to \infty} \phi'(C) = 0$,
		requiring the numerator also to vanish,
		which since $\kappa' > 0$ on $(0,\infty)$ demands that $L^\star(u) \to 0$.
	\end{proof}%
	\renewcommand{\qedsymbol}{$\blacksquare$}

	$F^0$ is strictly concave since $\phi$ is.
	Its (right-hand) derivative
	\begin{equation*}
		F^{0\prime}(u)
		= 1 - \frac{ \lambda }{ \phi'\left( \phi^{-1}(u) \right) }
	\end{equation*}
	is strictly positive at $u=0$ since $\lim_{C \to 0} \phi'(C) = \infty$,
	and is strictly negative for $u$ large enough since
	$\phi^{-1}(u) \to \infty$ as $u \to \infty$
	and $\lim_{C \to \infty} \phi'(C) = 0$.
	Thus $F^0$ is uniquely maximised at $u^0 \in (0,\infty)$ satisfying the first-order condition
	\begin{equation*}
		\phi'\left( \phi^{-1}\left( u^0 \right) \right) = \lambda .
	\end{equation*}

	$F^1$ is also strictly concave: for $u \neq u'$ in $[0,\infty)$ and $\eta \in (0,1)$,
	we have
	\begin{multline*}
		F^1( \eta u + (1-\eta) u' )
		\\
		\begin{aligned}
			&= \eta u + (1-\eta) u'
			+ \lambda \max_{L \geq 0} \left\{
			w L - \phi^{-1} \left( \eta u + (1-\eta) u' + \kappa(L) \right)
			\right\}
			\\
			&> \eta u + (1-\eta) u'
			\\
			&\quad
			+ \lambda \max_{L \geq 0} \left\{
			w L
			- \eta \phi^{-1} \left( u + \kappa(L) \right)
			- (1-\eta) \phi^{-1} \left( u' + \kappa(L) \right)
			\right\}
			\\
			&\geq \eta \left[ u
			+ \lambda \max_{L \geq 0} \left\{
			w L
			- \phi^{-1} \left( u + \kappa(L) \right)
			\right\}
			\right]
			\\
			&\quad
			+ (1-\eta) \left[ u'
			+ \lambda \max_{L \geq 0} \left\{
			w L
			- \phi^{-1} \left( u' + \kappa(L) \right)
			\right\}
			\right]
			\\
			&= \eta F^1(u) + (1-\eta) F^1(u') ,
		\end{aligned}
	\end{multline*}
	where the strict inequality holds since $-\phi^{-1}$ is strictly concave.
	By the envelope theorem, we have
	\begin{equation*}
		F^{1\prime}(u)
		= 1 - \frac{ \lambda }
		{ \phi'\left( \phi^{-1}\left(
		u + \kappa\left( L^\star(u) \right)
		\right) \right) } .
	\end{equation*}
	This expression is strictly negative for $u$ large enough
	because the denominator in the fraction vanishes as $u \to \infty$
	since $L^\star(u) \to 0$ by the \hyperref[claim:Lstar_interior]{claim}.
	Thus $F^1$ has a unique maximiser $u^1 \in [0,\infty)$,
	which satisfies the first-order condition
	\begin{equation*}
		\phi'\left( \phi^{-1} \left(
		u^1 + \kappa\left( L^\star\left(u^1\right) \right)
		\right) \right)
		\leq \lambda ,
		\quad \text{with equality if $u^1>0$.}
	\end{equation*}

	To show that $u^0 > u^1$, consider two cases.
	If $u^1 = 0$, then $u^0 > 0 = u^1$.
	If instead $u^1 > 0$, then the first-order conditions for $u^0$ and $u^1$ together yield
	\begin{equation*}
		\phi'\left( \phi^{-1}\left( u^0 \right) \right)
		= \lambda
		= \phi'\left( \phi^{-1} \left(
		u^1 + \kappa\left( L^\star\left(u^1\right) \right)
		\right) \right) .
	\end{equation*}
	Thus
	\begin{equation*}
		u^0 = u^1 + \kappa\left( L^\star\left(u^1\right) \right) > u^1 
	\end{equation*}
	since $L^\star\left( u^1 \right) > 0$ by the \hyperref[claim:Lstar_interior]{claim}
	and $\kappa > 0$ on $(0,\infty)$.

	It remains only to show that $F^1 - F^0$ is strictly decreasing, so that $u^\star = 0$.
	It suffices to show that $F^{1\prime} < F^{0\prime}$ on $(0,\infty)$.
	So fix any $u>0$.
	Since $L^\star(u) > 0$ by the \hyperref[claim:Lstar_interior]{claim},
	and $\kappa > 0$ on $(0,\infty)$, we have
	\begin{equation*}
		F^{0\prime}(u)
		= 1 - \frac{ \lambda }{ \phi'\left( \phi^{-1}(u) \right) }
		>  1 - \frac{ \lambda }
		{ \phi'\left( \phi^{-1}\left(
		u + \kappa\left( L^\star(u) \right)
		\right) \right) }
		= F^{1\prime}(u) .
		\qedhere
	\end{equation*}	
\end{proof}

\section*{\Cref*{suppl:ext:drop_F1-F0_strict} (p. \pageref*{suppl:ext:drop_F1-F0_strict})}
\label{omit:ext:drop_F1-F0_strict}
\addcontentsline{toc}{section}{Supplemental appendix \ref*{suppl:ext:drop_F1-F0_strict} (p. \pageref*{suppl:ext:drop_F1-F0_strict})}

To formalise and prove the trichotomy asserted in \cref*{footnote:trichotomy} (p. \pageref*{footnote:trichotomy}),
consider frontiers $F^0,F^1$ satisfying all of our model assumptions,
except for the requirement that $u^\star$ be a strict local maximum of $F^1-F^0$.

\begin{definition}
	\label{definition:saddle}
	$\widebar{u} \in (0,\infty)$ is a \emph{saddle point} of a continuous function $\psi$ defined on $[0,\infty)$ iff both
	\begin{enumerate}[label=(\alph*)]

		\item \label{bullet:saddle:deriv}
		for all $\eps > 0$, there are $u,u' \in (0,\infty)$ with $\widebar{u}-\eps < u < \widebar{u} < u' < \widebar{u}+\eps$ such that $\abs*{ \psi(u')-\psi(u) } / (u'-u) < \eps$,%
			\footnote{In case $\psi$ is differentiable at $\widebar{u}$,
			condition \ref{bullet:saddle:deriv} reduces to the requirement that $\psi'\left(\widebar{u}\right)=0$.
			Our definition thus extends the usual definition of a saddle point, which applies only to differentiable functions.}
		and

		\item \label{bullet:saddle:maxmin}
		$\widebar{u}$ is neither a local maximum nor a local minimum of $\psi$.

	\end{enumerate}
\end{definition}

\begin{namedthm}[Claim.]
	If $u^\star>0$, then either
	(i) $u^\star$ is a local maximum of $F^1-F^0$,
	(ii) $u^\star$ is a saddle point of $F^1-F^0$, or
	(iii) both $F^0$ and $F^1$ are kinked at $u^\star$.
\end{namedthm}

Assuming that $u^\star$ is a strict local maximum of $F^1-F^0$
thus amounts to ruling out the saddle-point case (ii),
the (pathological) mutual-kink case (iii),
and the case in which $F^1-F^0$ is constant on a left-neighbourhood of $u^\star$.

\begin{proof}
	Assume that $u^\star>0$ and that neither (i) nor (ii) holds;
	we shall deduce (iii).
	By its definition, $u^\star$ cannot be a local minimum of $F^1-F^0$.
	Then by definition of `saddle point',
	there must be an $\eps>0$ such that
	\begin{equation*}
		\frac{ \abs*{ \left(F^1-F^0\right)(u') - \left(F^1-F^0\right)(u) } }
		{u'-u}
		\geq \eps
	\end{equation*}
	for all $u,u' \in (0,\infty)$ satisfying $u^\star-\eps < u < u^\star < u' < u^\star+\eps$.
	Since $F^1-F^0$ is strictly decreasing on $\left[u^\star,u^0\right]$,
	the expression in the numerator must be strictly negative for $u$ sufficiently close to $u^\star$,
	and thus $F^{1+}(u^\star) \leq F^{0+}(u^\star) - \eps < F^{0+}(u^\star)$ and (similarly) $F^{1-}(u^\star) < F^{0-}(u^\star)$.
	By definition of $u^\star$, the frontiers $F^0,F^1$ share a supergradient $\eta \in \R$ at $u^\star$,
	so that $F^{j+}(u^\star) \leq \eta \leq F^{j-}(u^\star)$ for both $j \in \{0,1\}$.
	Then
	\begin{equation*}
		F^{1+}(u^\star) < F^{0+}(u^\star)
		\leq \eta
		\leq F^{1-}(u^\star) < F^{0-}(u^\star) ,
	\end{equation*}
	implying that both $F^0$ and $F^1$ are kinked at $u^\star$.
\end{proof}

\section*{\Cref*{suppl:ext:random_F1} (p. \pageref*{suppl:ext:random_F1})}
\label{omit:ext:random_F1}
\addcontentsline{toc}{section}{Supplemental appendix \ref*{suppl:ext:random_F1} (p. \pageref*{suppl:ext:random_F1})}

For $u \in [0,\infty)$,
write $M_u$ for the space of all maps $\widehat{X} : \mathcal{F} \to [0,\infty)$ that satisfy
$\E\bigl( \widehat{X}(\boldsymbol{F}) \bigr) = u$.

\begin{proof}[Proof of \Cref*{lemma:randomF1_assns}]
	For concavity, take $u,u^\dag \in [0,\infty)$ and $\lambda \in (0,1)$.
	Let $\widehat{X} \in M_u$ and $\widehat{X}^\dag \in M_{u^\dag}$ be maximisers:
	\begin{equation*}
		\E\left( \boldsymbol{F}\left( \widehat{X}( \boldsymbol{F} ) \right) \right)
		= F^1(u)
		\quad \text{and} \quad
		\E\left( \boldsymbol{F}\left( \widehat{X}^\dag( \boldsymbol{F} ) \right) \right)
		= F^1\left( u^\dag \right) .
	\end{equation*}
	Then $\lambda \widehat{X} + (1-\lambda) \widehat{X}^\dag$ belongs to $M_{\lambda u + (1-\lambda) u^\dag}$, so that
	\begin{multline*}
		F^1\left( \lambda u + (1-\lambda) u^\dag \right)
		\\
		\begin{aligned}
			&\geq \E\left(
			\boldsymbol{F}\left(
			\lambda \widehat{X}( \boldsymbol{F} )
			+ (1-\lambda) \widehat{X}^\dag( \boldsymbol{F} )
			\right)
			\right)
			&&
			\\
			&\geq \E\left(
			\lambda \boldsymbol{F}\left( \widehat{X}( \boldsymbol{F} ) \right)
			+ (1-\lambda) \boldsymbol{F}\left( \widehat{X}^\dag( \boldsymbol{F} )\ \right)
			\right)
			&& \text{since $\boldsymbol{F}$ is concave a.s.}
			\\
			&= \lambda F^1(u) + (1-\lambda) F^1\left( u^\dag \right) .
			&&
		\end{aligned}
	\end{multline*}

	To see that $F^1$ attains a maximum at $u^1$, observe that $F^1$ is bounded above by $\E\left( \boldsymbol{F}\left( U^1( \boldsymbol{F} ) \right) \right)$,
	and that it attains this value at $u^1$.
	The peak is unique because for $u \neq u^1$, any $\widehat{X} \in M_u$ must have $\widehat{X} \neq U^1$ on a non-null set, and thus
	\begin{equation*}
		\boldsymbol{F}\left( \widehat{X}(\boldsymbol{F}) \right)
		\leq \mathrel{(<)} \boldsymbol{F}\left( U^1(\boldsymbol{F}) \right)
		\quad \text{a.s. (with positive probability).}
	\end{equation*}

	For upper semi-continuity, observe that since $F^1$ is concave,
	it is continuous on the interior of its effective domain, which is a convex set that contains $u^1$.
	$F^1$ is trivially continuous off the closure of its effective domain, where it is constant and equal to $-\infty$.
	It remains only to establish that
	\begin{equation*}
		\limsup_{u' \to u} F^1(u') \leq F^1(u)
	\end{equation*}
	for $u$ on the boundary of the effective domain (there are at most two such $u$s).
	It suffices to show for an arbitrary decreasing sequence $(u_n)_{n \in \N}$ in the interior of the effective domain converging to some $u \in \left[ 0, u^1 \right]$ that
	\begin{equation*}
		\lim_{n \to \infty} F^1(u_n) \leq F^1(u)
	\end{equation*}
	(where the limit exists since $\left( F^1(u_n) \right)_{n \in \N}$ is eventually monotone),
	and similarly for increasing sequences converging to $u \geq u^1$.
	We show the former, omitting the analogous argument for the latter.

	For each $n \in \N$, let $\widehat{X}_n \in M_{u_n}$ be a maximiser at $u_n$:
	\begin{equation*}
		\E\left( \boldsymbol{F}\left( \widehat{X}_n( \boldsymbol{F} ) \right) \right)
		= F^1(u_n) .
	\end{equation*}
	Since $\Union_{ u' \in \left[ 0, u^1 \right] } M_{u'}$ is compact (because bounded),
	the sequence $\bigl( \widehat{X}_n \bigr)_{n \in \N}$ admits a convergent subsequence $\bigl( \widehat{X}_{n_k} \bigr)_{k \in \N}$, whose limit we denote by $\widehat{X}$.
	We have
	\begin{align*}
		F^1(u)
		&\geq \E\left( \boldsymbol{F}\left( \widehat{X}( \boldsymbol{F} ) \right) \right)
		&& \text{since $\widehat{X} \in M_u$}
		\\
		&\geq \E\left( \lim_{k \to \infty} \boldsymbol{F}\left( \widehat{X}_{n_k}( \boldsymbol{F} ) \right) \right)
		&& \text{since $\boldsymbol{F}$ is upper semi-continuous a.s.}
		\\
		&= \lim_{k \to \infty}
		\E\left( \boldsymbol{F}\left( \widehat{X}_{n_k}( \boldsymbol{F} ) \right) \right)
		&& \text{by bounded convergence}
		\\
		&= \lim_{k \to \infty} F^1\left( u_{n_k} \right)
		&&
		\\
		&= \lim_{n \to \infty} F^1\left( u_n \right)
		&& \text{since $\left( F^1(u_n) \right)_{n \in \N}$ is convergent.}
		\qedhere
	\end{align*}	
\end{proof}

\section*{\Cref*{suppl:Euler_lemma_pf} (p. \pageref*{suppl:Euler_lemma_pf})}
\label{omit:Euler_lemma_pf}
\addcontentsline{toc}{section}{Supplemental appendix \ref*{suppl:Euler_lemma_pf} (p. \pageref*{suppl:Euler_lemma_pf})}

To prove \Cref*{lemma:euler_integrable_supergradient} and the Gateaux lemma, we shall use the following standard integration-by-parts result:%
	\footnote{See e.g. Theorem 18.4 in \textcite[p. 236]{Billingsley1995}.}

\begin{namedthm}[IBP lemma.]
	\label{lemma:int_parts}
	Let $\nu$ be a finite measure on $\R_+$, and let $L$ be a $\nu$-integrable function $\R_+ \to \R$ satisfying $L(t) = L(0) + \int_0^t l$ for some (Lebesgue-)integrable $l : \R_+ \to \R$.
	Then
	\begin{equation*}
		\int_{[0,T]} L \dd \nu = L(T) \nu([0,T]) - \int_0^T \nu([0,t]) l(t) \dd t
		\quad \text{for every $T \in \R_+$.}
	\end{equation*}
\end{namedthm}

\subsection{Proof of Lemma \ref*{lemma:euler_integrable_supergradient} (p. \pageref*{lemma:euler_integrable_supergradient})}
\label{omit:Euler_lemma_pf:euler_integrable_supergradient}

Define $\Phi : \R_+ \to [0,\infty]$ by
\begin{equation*}
	\Phi(t) \coloneqq r \int_0^t e^{-rs} \phi^0(s) \dd s .
\end{equation*}
We first show that $\E_G(\Phi(\tau))<\infty$,
so that $\Phi$ is $G$-integrable.
Note that
\begin{align*}
	\E_G\left( r \int_0^\tau \frac{e^{-rt}}{1-G(t)} \dd t \right)
	&= \lim_{T \to \infty} \int_{[0,T]}
	\left( r \int_0^s \frac{e^{-rt}}{1-G(t)} \dd t \right)
	G(\dd s)
	\\
	&= \lim_{T \to \infty} r \int_0^T e^{-rt}\frac{G(T)-G(t)}{1-G(t)}\dd t
	= r \int_0^\infty e^{-rt} \dd t
	= 1 ,
\end{align*}
where the first and third equalities hold by monotone convergence, and the second follows from the \hyperref[lemma:int_parts]{IBP lemma}
with $\nu$ the measure associated with $G$
and
\begin{equation*}
	l(t) \coloneqq
	\begin{cases}
		r e^{-rt} / [1-G(t)]	& \text{for $t < T$} \\
		0						& \text{for $t \geq T$}
	\end{cases}
	\qquad \text{if $G(T)<1$}
\end{equation*}
and $l(t) \coloneqq r e^{-rt}$ if $G(T)=1$.
Thus since $(x,X)$ satisfies the Euler equation with $\phi^0,\phi^1$,
(\ref*{eq:euler}) yields
\begin{align*}
	\E_G\left( \Phi(\tau) \right)
	&= \E_G\left( r \int_0^\tau e^{-rt}
	\frac{-\int_{[0,t]}\phi^1 \dd G}{1-G(t)} \dd t\right)
	\\
	&\leq \E_G\left( r \int_0^\tau e^{-rt}
	\frac{\int_{\R_+}\abs*{\phi^1} \dd G}{1-G(t)} \dd t\right)
	= \E_G\left(\abs*{\phi^1(\tau)}\right) 
	< \infty ,
\end{align*}
where the final inequality holds since $\phi^1$ is $G$-integrable.

To show that $x$ belongs to $\mathcal{X}_G$,
fix a $u \in \left(0,u^0\right)$;
we shall prove that the maps $\psi^0_{x,u}$ and $\psi^1_{X,u}$ (defined on p. \pageref*{eq:euler_psi01_defn}) are $G$-integrable.
For the former, let $\mathcal{T}$ be the set of $t \in \R_+$ at which $x_t \geq u$.
The map $t \mapsto F^{0\prime}(x_t,u)$ is bounded on $\mathcal{T}$
(by zero below and by $F^{0+}(u)$ above) since $x \leq u^0$ and $F^0$ is increasing and concave,
so it suffices to show that
\begin{equation*}
	\varphi(t)
	\coloneqq r \int_0^t e^{-rs}
	F^{0\prime}\left(x_s,u\right) \1_{\R_+ \setminus \mathcal{T}}(s) \dd s
\end{equation*}
is $G$-integrable.
Since $(x,X)$ satisfies the Euler equation, $\phi^0(t)$ is a supergradient of $F^0$ at $x_t$ for a.e. $t \in \R_+$ such that $G(t) < 1$.
Thus since $x < u$ on $\R_+ \setminus \mathcal{T}$ and $F^0$ is concave, we have
\begin{equation*}
	\phi^0(t)
	\geq F^{0+}(x_t)
	= F^{0\prime}(x_t,u)
	\quad \text{for a.e. $t \in \R_+ \setminus \mathcal{T}$ with $G(t)<1$,}
\end{equation*}
so that
\begin{equation*}
	\varphi(t)
	\leq r \int_0^t e^{-rs} \phi^0(s) \1_{\R_+ \setminus \mathcal{T}}(s) \dd s
	\leq \Phi(t)
	\quad \text{for all $t \in \R_+$ with $G(t)<1$.}
\end{equation*}
Since $\varphi$ and $\Phi$ are continuous
and the former is non-negative,
it follows that $0 \leq \varphi \leq \Phi$ on the support of $G$,
so that $\varphi$ is $G$-integrable since $\Phi$ is.

It remains to show that $\psi^1_{X,u}$ is $G$-integrable.
Choose $\eps \in \left(0,u \meet \left[u^0 - u \right]\right)$ so that $\eps < u^1 \meet \left(u^0 - u^1 \right)$ if $u^1 > 0$,
and let
\begin{equation*}
	\mathcal{T}'
	\coloneqq
	\begin{cases}
		\left\{t \in \R_+ : X_t < u + \eps\right\}
		& \text{if $u^1 = 0$}
		\\
		\left\{t \in \R_+ : \left(u \meet u^1\right)-\eps < X_t < \left(u \join u^1\right) + \eps\right\}
		& \text{if $u^1 > 0$.}
	\end{cases}
\end{equation*}
We will show that $t \mapsto \psi^1_{X,u}(t) = e^{-rt} F^{1\prime}(X_t,u)$
is bounded on $\mathcal{T}'$
and $G$-integrable on $\R_+ \setminus \mathcal{T}'$.
The former follows from the fact that
\begin{equation*}
	\text{$F^1$ is Lipschitz continuous on}\quad
	\begin{cases}
		\left[0,u+\eps\right]
		& \text{if $u^1=0$} \\
		\left[ \left(u \meet u^1\right)-\eps,
		\left(u \join u^1\right) + \eps\right]
		& \text{if $u^1>0$,}
	\end{cases}
\end{equation*}
so that $t \mapsto F^{1\prime}(X_t,u)$ is bounded on $\mathcal{T}'$.
For the latter,
note that (by definition of $\mathcal{T}'$,)
every $t \in \R_+ \setminus \mathcal{T}'$ has
either $X_t < u, u^1$ or $X_t > u,u^1$.
Since $(x,X)$ satisfies the Euler equation, $\phi^1(t)$ is a supergradient of $F^1$ at $X_t$ for $G$-a.e. $t \in \R_+$.
Thus for $G$-a.e. $t \in \R_+ \setminus \mathcal{T}'$, we have
\begin{align*}
	\text{either}\quad
	F^{1\prime}(X_t,u) = F^{1+}(X_t)
	\quad &\text{and} \quad
	0 < F^{1+}(X_t) \leq \phi^1(t)
	\\
	\text{or}\quad
	F^{1\prime}(X_t,u) = F^{1-}(X_t)
	\quad &\text{and} \quad
	0 > F^{1-}(X_t) \geq \phi^1(t) 
\end{align*}
by the concavity of $F^0$,
so that
\begin{equation*}
	\abs*{\psi^1_{X,u}(t)}
	\leq \abs*{F^{1\prime}(X_t,u)}
	\leq \abs*{\phi^1(t)}
	\quad \text{for $G$-a.e. $t \in \R_+ \setminus \mathcal{T}'$,}
\end{equation*}
which implies that $\psi^1_{X,u}$ is $G$-integrable on $\R_+ \setminus \mathcal {T}'$
since $\phi^1$ is.
\qed

\subsection{Proof of the claim in the proof of Lemma \ref*{lemma:integrable_supergradient} (p. \pageref*{claim:integrable_supergradient})}
\label{omit:Euler_lemma_pf:integrable_supergradient}

\renewcommand{\qedsymbol}{$\square$}
For $\alpha \in (0,1)$, define $h^0_\alpha, h^1_\alpha : \R_+ \to (-\infty,\infty]$ by
\begin{align*}
	h^0_\alpha(t)
	&\coloneqq r e^{-rt}
	\frac{ F^0\left( x_t+\alpha \left[ u - x_t \right] \right)
	- F^0( x_t ) }
	{ \alpha }
	\\
	\text{and}\quad
	h^1_\alpha(t)
	&\coloneqq e^{-rt}
	\frac{ F^1\left( X_t + \alpha \left[u-X_t\right] \right) - F^1(X_t) }
	{\alpha} .
\end{align*}
Clearly both are measurable for each $\alpha \in (0,1)$.
By definition, we have
\begin{equation*}
	\DD \pi_G\bigl( x, x^\dag-x \bigr)
	= \lim_{\alpha \downarrow 0} \left[
	\E_G\left( \int_0^\tau h^0_\alpha \right)
	+ \E_G\left( h^1_\alpha(\tau) \right)
	\right] ;
\end{equation*}
we must show that the limit exists and satisfies the asserted inequality.

As $\alpha \downarrow 0$,
$h^0_\alpha$ and $h^1_\alpha$
converge pointwise
to
$h^0 : \R_+ \to (-\infty,\infty]$
and $h^1 : \R_+ \to [-\infty,\infty]$, respectively, given by
\begin{equation*}
	h^0(t)
	\coloneqq r e^{-rt} F^{0\prime}\left( x_t, u \right)
	\left( u - x_t \right)
	\quad \text{and} \quad
	h^1(t)
	\coloneqq e^{-rt} F^{1\prime}\left( X_t, u \right)
	\left( u - X_t \right) .
\end{equation*}
We shall show that
\begin{enumerate}[label=(\alph*)]

	\item \label{bullet:integrable_supergradient_a}
	$t \mapsto \int_0^t h^0 \1_{\mathcal{T}}$
	and $h^1 \1_{\mathcal{T}'}$ are $G$-integrable,

	\item \label{bullet:integrable_supergradient_b}
	$\E_G\left( \int_0^\tau h^0 \right)$ and $\E_G\left( h^1(\tau) \right)$ exist and are $>-\infty$, and

	\item \label{bullet:integrable_supergradient_c}
	$\E_G\left( \int_0^\tau h^0_\alpha \right)
	\to \E_G\left( \int_0^\tau h^0 \right)$
	and $\E_G\left( h^1_\alpha(\tau) \right)
	\to \E_G\left( h^1(\tau) \right)$
	as $\alpha \downarrow 0$.

\end{enumerate}
This suffices since then
$C \coloneqq
  \E_G\left( \int_0^\tau h^0 \1_{\mathcal{T}} \right)
+ \E_G\left( h^1(\tau) \1_{\mathcal{T}'}(\tau) \right)$ is finite,
\begin{align*}
	\DD \pi_G\bigl( x, x^\dag-x \bigr)
	&= \E_G\left( \int_0^\tau h^0 \right)
	+ \E_G\left( h^1(\tau) \right)
	\\
	&= \E_G\left( \int_0^\tau h^0 \1_{\R_+ \setminus \mathcal{T}} \right)
	+ \E_G\left( h^1(\tau) \1_{\R_+ \setminus \mathcal{T}'}(\tau) \right)
	+ C ,
\end{align*}
and moreover any $t \in \R_+ \setminus \mathcal{T}$ has
(by definition of $\mathcal{T}$)
\begin{equation*}
	u-x_t \geq \eps
	\quad \text{and thus} \quad
	h^0(t) \geq r e^{-rt} F^{0\prime}(x_t,u) \eps
	= r e^{-rt} \abs*{ F^{0\prime}(x_t,u) } \eps ,
\end{equation*}
while any $t \in \R_+ \setminus \mathcal{T}'$ has
\begin{align*}
	\text{either}\quad
	u-X_t \geq \eps'
	\quad &\text{and} \quad
	X_t < u^1
	\quad \text{(hence $F^{1\prime}\left(X_t,u\right) > 0$)}
	\\
	\text{or (if $u^1 > 0$)}\quad
	u-X_t \leq \eps'
	\quad &\text{and} \quad
	X_t > u^1
	\quad \text{(hence $F^{1\prime}\left(X_t,u\right) < 0$)} 
\end{align*}
by definition of $\mathcal{T}'$,
so that $h^1(t) \geq e^{-rt}\abs*{F^{1\prime}\left(X_t,u\right)} \eps'$.

To obtain \ref{bullet:integrable_supergradient_a}--\ref{bullet:integrable_supergradient_c} for $h^0$,
observe that $F^0$ is $K$-Lipschitz on $\left[ u-\eps, u^0 \right]$ for some $K>0$ since it is concave and increasing on $\left[0,u^0\right]$.
Thus
\begin{equation*}
	\abs*{ h^0_\alpha(t) }
	\leq r e^{-rt} K \abs*{ x_t - u }
	\leq r e^{-rt} K u^0
	\quad \text{for all $\alpha \in (0,1)$ and $t \in \mathcal{T}$,}
\end{equation*}
which implies that
\begin{equation*}
	\abs*{ \int_0^t h^0_\alpha \1_{\mathcal{T}} }
	\leq \int_0^t \abs*{ h^0_\alpha } \1_{\mathcal{T}}
	\leq K u^0 r \int_0^t e^{-rt} \1_{\mathcal{T}}
	\leq K u^0
	\quad \text{for all $\alpha \in (0,1)$.}
\end{equation*}
Thus applying the dominated convergence theorem twice yields
\begin{equation*}
	\lim_{\alpha \downarrow 0}
	\E_G\left( \int_0^\tau h^0_\alpha \1_{\mathcal{T}} \right)
	= \E_G\left(
	\lim_{\alpha \downarrow 0}
	\int_0^\tau h^0_\alpha \1_{\mathcal{T}} \right)
	= \E_G\left( \int_0^\tau h^0 \1_{\mathcal{T}} \right)
	\in \R ,
\end{equation*}
giving us property \ref{bullet:integrable_supergradient_a}.

Meanwhile,
$h^0_\alpha \1_{\R_+ \setminus \mathcal{T}}$
is (pointwise) decreasing in $\alpha$ since $F^0$ is concave,
and is non-negative since $F^0$ is increasing
on $\left[ 0, u^0 \right]$ and $x \leq u < u^0$ on $\R_+ \setminus \mathcal{T}$.
Clearly the same is true of $t \mapsto \int_0^t h^0_\alpha \1_{\R_+ \setminus \mathcal{T}}$,
Thus applying the monotone convergence theorem twice yields
\begin{equation*}
	0
	\leq \lim_{\alpha \downarrow 0}
	\E_G\left( \int_0^\tau h^0_\alpha \1_{\R_+ \setminus \mathcal{T}} \right)
	= \E_G\left(
	\lim_{\alpha \downarrow 0}
	\int_0^\tau h^0_\alpha \1_{\R_+ \setminus \mathcal{T}} \right)
	= \E_G\left( \int_0^\tau h^0 \1_{\R_+ \setminus \mathcal{T}} \right) ,
\end{equation*}
where the rightmost term is well-defined.
Thus $h^0$ satisfies \ref{bullet:integrable_supergradient_b} and \ref{bullet:integrable_supergradient_c}.

To derive \ref{bullet:integrable_supergradient_a}--\ref{bullet:integrable_supergradient_c} for $h^1$,
observe that
\begin{equation*}
	\text{$F^1$ is Lipschitz continuous on}\quad
	\begin{cases}
		\left[0,u+\eps'\right]
		& \text{if $u^1=0$} \\
		\left[ \left(u \meet u^1\right)-\eps',
		\left(u \join u^1\right) + \eps'\right]
		& \text{if $u^1>0$}
	\end{cases}
\end{equation*}
since it is concave and maximised at $u^0$.
Hence the family $\smash{\left( h^1_\alpha \right)_{\alpha \in (0,1)}}$ is uniformly bounded on $\mathcal{T}'$, so that
\begin{equation}
	\lim_{\alpha \downarrow 0}
	\E_G\left( h^1_\alpha(\tau) \1_{\mathcal{T}'}(\tau) \right)
	= \E_G\left( h^1(\tau) \1_{\mathcal{T}'}(\tau) \right)
	\in \R
	\label{eq:integrable_supergradient_1T}
\end{equation}
by bounded convergence.
Thus $h^1$ satisfies property \ref{bullet:integrable_supergradient_a}.

Define $\widebar{\alpha} \coloneqq \eps' / \left( \eps' + \abs*{u^1-u} \right)$.
For any $t \in \R_+ \setminus \mathcal{T}'$ we have
either $X_t \leq \left( u \meet u^1 \right) - \eps' < u$
or (if $u^1>0$) $X_t \geq \left( u \join u^1 \right) + \eps' > u$.
Thus for any $\alpha \in \left( 0, \widebar{\alpha} \right)$,
either
\begin{equation*}
	X_t
	< (1-\alpha) X_t + \alpha u
	\leq (1-\alpha) \left( u \meet u^1 \right) - (1-\alpha) \eps' + \alpha u
	< u^1
\end{equation*}
(where the last inequality holds by $\alpha<\widebar{\alpha}$)
or (similarly)
\begin{equation*}
	X_t
	> (1-\alpha) X_t + \alpha u
	\geq (1-\alpha) \left( u \join u^1 \right) + (1-\alpha) \eps' + \alpha u
	> u^1 .
\end{equation*}
In short, $(1-\alpha) X_t + \alpha u$ lies between $X_t$ and $u^1$
for any $t \in \R_+ \setminus \mathcal{T}'$ and $\alpha \in \left( 0, \widebar{\alpha} \right)$,
which since $F^1$ is concave and maximised at $u^1$ implies that
$\left( 0, \widebar{\alpha} \right) \ni \alpha \mapsto h^1_\alpha \1_{\R_+ \setminus \mathcal{T}'}$ is (pointwise) decreasing and non-negative.
Hence
\begin{equation*}
	0
	\leq \lim_{\alpha \downarrow 0}
	\E_G\left( h^1_\alpha(\tau) \1_{\R_+ \setminus \mathcal{T}'}(\tau) \right)
	= \E_G\left( h^1(\tau) \1_{\R_+ \setminus \mathcal{T}'}(\tau) \right)
\end{equation*}
by the monotone convergence theorem,
where the right-hand side is well-defined.
This together with \eqref{eq:integrable_supergradient_1T}
implies that $h^1$ satisfies \ref{bullet:integrable_supergradient_b} and \ref{bullet:integrable_supergradient_c}.
\qed
\renewcommand{\qedsymbol}{$\blacksquare$}

\subsection{Proof of the Gateaux lemma (p. \pageref*{lemma:euler_gateaux})}
\label{omit:Euler_lemma_pf:euler_gateaux}

For $\alpha \in (0,1)$, define $h^0_\alpha, h^1_\alpha : \R_+ \to (-\infty,\infty]$:
\begin{align*}
	h^0_\alpha(t)
	&\coloneqq r e^{-rt}
	\frac{ F^0\left( x_t+\alpha \left[ x^\dag_t - x_t \right] \right)
	- F^0( x_t ) }
	{ \alpha }
	\\
	\text{and}\quad
	h^1_\alpha(t)
	&\coloneqq e^{-rt}
	\frac{ F^1\left( X_t + \alpha \left[X^\dag_t-X_t\right] \right) - F^1(X_t) }
	{\alpha} .
\end{align*}
Clearly both are measurable for each $\alpha \in (0,1)$.
We have
\begin{equation*}
	\DD \pi_G\bigl( x, x^\dag-x \bigr)
	= \lim_{\alpha \downarrow 0} \left[
	\E_G\left( \int_0^\tau h^0_\alpha \right)
	+ \E_G\left( h^1_\alpha(\tau) \right)
	\right]
\end{equation*}
by definition;
we must show that the limit exists and has the asserted value.

As $\alpha \downarrow 0$,
$h^0_\alpha$ and $h^1_\alpha$
converge pointwise
to
$h^0 : \R_+ \to (-\infty,\infty]$
and $h^1 : \R_+ \to [-\infty,\infty]$, respectively, given by
\begin{align*}
	h^0(t)
	&\coloneqq r e^{-rt} F^{0\prime}\left( x_t, x^\dag_t \right)
	\left( x^\dag_t - x_t \right)
	\\
	\text{and}\quad
	h^1(t)
	&\coloneqq e^{-rt} F^{1\prime}\left( X_t, X^\dag_t \right)
	\left( X^\dag_t - X_t \right) .
\end{align*}

\begin{namedthm}[Claim A.]
	\label{claim:gateaux_integral_convergence}
	$\DD\pi_G\bigl(x,x^\dag-x\bigr)$ exists
	and is finite, equal to
	$\E_G\left( \int_0^\tau h^0 \right)
	+ \E_G\left( h^1(\tau) \right)$.
\end{namedthm}

By \hyperref[claim:gateaux_integral_convergence]{claim A}, we have
\begin{multline*}
	\DD \pi_G\left( x, x^\dag-x \right)
	= \E_G\left( \int_0^\tau h^0 \right)
	+ \E_G\left(h^1(\tau)\right)
	\\
	\begin{aligned}
		&= \E_G\left( r \int_0^\tau e^{-rs} \phi^0(s)
		\left[ x_s^\dag - x_s \right] \dd s \right)
		+ \E_G\left( e^{-r\tau} \phi^1(\tau)
		\left[ X^\dag_\tau - X_\tau \right] \right)
		\\
		&\quad
		+ \E_G\left( r \int_0^\tau e^{-rt}
		\left[ F^{0\prime}\left( x_t, x^\dag_t \right) - \phi^0(t) \right]
		\left[ x_t - x^\dag_t \right] \dd t \right)
		\\
		&\quad
		+ \E_G\left( e^{-r\tau}
		\left[ F^{1\prime}\left( X_\tau, X^\dag_\tau \right) - \phi^1(\tau) \right]
		\left[ X^\dag_\tau - X_\tau \right] \right) ,
	\end{aligned}
\end{multline*}
where the last equality holds since $t \mapsto r \int_0^t e^{-rs}\phi^0(s) \dd s$ and $\phi^1$ are $G$-integrable
and $x,x^\dag,X,X^\dag$ are bounded.
The proof is completed by invoking the following claims to rewrite the first two terms in the desired form:

\begin{namedthm}[Claim B.]
	\label{claim:gateaux_algebra_b}
	It holds that
	\begin{equation*}
		\E_G\left( r \int_0^\tau e^{-rs}\phi^0(s)
		\left[ x_s^\dag - x_s \right] \dd s
		\right)
		= r \int_0^\infty e^{-rt} [1-G(t)]\phi^0(t)
		\left( x_t^\dag - x_t \right) \dd t .
	\end{equation*}
\end{namedthm}

\begin{namedthm}[Claim C.]
	\label{claim:gateaux_algebra_c}
	It holds that
	\begin{equation*}
		\E_G\left(
		e^{-r\tau}\phi^1(\tau) \left[ X^\dag_\tau - X_\tau \right]
		\right)
		= r \int_0^\infty e^{-rt}
		\left( \int_{[0,t]} \phi^1 \dd G \right)
		\left( x^\dag_t - x_t \right) \dd t .
	\end{equation*}
\end{namedthm}

\begin{proof}[Proof of {\hyperref[claim:gateaux_integral_convergence]{claim A}}]%
	\renewcommand{\qedsymbol}{$\square$}
	We must show that $t \mapsto \int_0^t h^0$ and $h^1$
	are $G$-integrable and
	\begin{equation*}
		\E_G\left( \int_0^\tau h^0_\alpha \right)
		\to \E_G\left( \int_0^\tau h^0 \right)
		\quad \text{and} \quad
		\E_G\left( h^1_\alpha(\tau) \right)
		\to \E_G\left( h^1(\tau) \right)
		\quad \text{as $\alpha \downarrow 0$.}
	\end{equation*}
	Fix a $u \in \left(0,u^0\right)$.

	For $h^0$, define $\varphi : \R_+ \to [0,\infty]$ by
	\begin{equation*}
		\varphi(t)
		\coloneqq r e^{-rt}
		\left[ F^{0\prime}(x_t,u) + F^{0\prime}\left(x_t^\dag,u\right) \right] .
	\end{equation*}
	Recall the definition of $\psi^0_{x,u},\psi^0_{x^\dag,u}$ (p. \pageref*{eq:euler_psi01_defn}).
	Since $x$ and $x^\dag$ belong to $\mathcal{X}_G$,
	the map
	\begin{equation*}
		t \mapsto \int_0^t \varphi = \psi^0_{x,u}(t) + \psi^0_{x^\dag,u}(t)
	\end{equation*}
	is $G$-integrable,
	and thus $\varphi$ (being non-negative) is Lebesgue-integrable on $(0,t)$ for any $t \in \R_+$ with $G(t)<1$.%
		\footnote{If $\varphi$ were not integrable on $(0,t)$,
		then $\smash{\psi^0_{x,u}(s) + \psi^0_{x^\dag,u}(s) = \int_0^s \varphi \geq \int_0^t \varphi = \infty}$ for every $s \geq t$,
		which if $G(t)<1$ implies
		$\smash{\E_G( \psi^0_{x,u}(\tau) + \psi^0_{x^\dag,u}(\tau) ) = \infty}$,
		a contradiction.}
	Since $F^0$ is concave, we have
	\begin{equation*}
		\abs*{ h^0_\alpha(t) }
		\leq \varphi(t) \abs*{ x_t - x^\dag_t }
		\leq \varphi(t) u^0
		\quad \text{for all $\alpha \in (0,1)$ and $t \in \R_+$,}
	\end{equation*}
	and thus
	\begin{equation*}
		\abs*{ \int_0^t h^0_\alpha }
		\leq \left[\psi^0_{x,u}(t) +\psi^0_{x^\dag,u}(t)\right] u^0
		\quad \text{for all $\alpha \in (0,1)$ and $t \in \R_+$.}
	\end{equation*}
	Thus applying the dominated convergence theorem twice yields
	\begin{equation*}
		\lim_{\alpha \downarrow 0}
		\E_G\left( \int_0^\tau h_\alpha^0 \right)
		= \E_G\left( \lim_{\alpha \downarrow 0}
		\int_0^\tau h_\alpha^0 \right)
		= \E_G\left(\int_0^\tau h^0 \right)
		\in \R .
	\end{equation*}

	Similarly, since $F^1$ is concave and $X,X^\dag$ take values in $\left[0,u^0\right]$, we have
	\begin{align*}
		\abs*{h^1_\alpha(t)}
		&\leq \left( \abs*{\psi^1_{X,u}(t)}
		+ \abs*{\psi^1_{X^\dag,u}(t)} \right)
		\abs*{X^\dag_t- X_t}
		\\
		&\leq \left(\abs*{\psi^1_{X,u}(t)}+\abs*{\psi^1_{X^\dag,u}(t)}\right) u^0
		&& \text{for all $\alpha \in (0,1)$ and $t \in \R_+$.}
	\end{align*}
	The right-hand side is $G$-integrable by since $x$ and $x^\dag$ belong to $\mathcal{X}_G$,
	so
	\begin{equation*}
		\lim_{\alpha \downarrow 0}
		\E_G\left( h^1_\alpha(\tau) \right)
		= \E_G\left( h^1(\tau) \right)
		\in \R
	\end{equation*}
	by dominated convergence.
\end{proof}%
\renewcommand{\qedsymbol}{$\blacksquare$}

\begin{proof}[Proof of {\hyperref[claim:gateaux_algebra_b]{claim B}}]%
	\renewcommand{\qedsymbol}{$\square$}
	Fix a $T \in \R_+$ such that $\lim_{t \uparrow T} G(t) < 1$, and note that $l_T : \R_+ \to \R$ given by
	\begin{equation*}
		l_T(t)
		\coloneqq r e^{-rt} \phi^0(t) \left[ x_t^\dag - x_t \right]
		\1_{[0,T)}(t)
	\end{equation*}
	is (Lebesgue-)integrable since $t \mapsto r \int_0^t e^{-rs}\phi^0(s) \dd s$ is $G$-integrable and $x,x^\dag$ are bounded.
	We may therefore apply \hyperref[lemma:int_parts]{IBP lemma} (p. \pageref{lemma:int_parts}) to $l_T$ and the measure associated with $G$ to obtain
	\begin{multline*}
		\int_{[0,T]} \left(
		r \int_0^t e^{-rs}\phi^0(s) \left[ x_s^\dag - x_s \right] \dd s
		\right) G( \dd t)
		\\
		= r \int_0^T [G(T)-G(t)]
		e^{-rt}\phi^0(t)
		\left[ x_t^\dag - x_t \right]
		\dd t .
	\end{multline*}

	Define $\widebar{T} \coloneqq \sup\{ t \in \R_+: G(t) < 1 \}$.
	If $\widebar{T} < \infty$
	and $G$ has an atom at $\widebar{T}$,
	then we obtain the desired equation by setting $T \coloneqq \widebar{T}$.
	Suppose for the remainder that this is not the case.
	It suffices to show that
	$\chi : \R_+ \to [0,\infty]$ given by
	\begin{equation*}
		\chi(t)
		\coloneqq
		\begin{cases}
			[1-G(t)] r e^{-rt} \phi^0(t)
			& \text{for $t < \widebar{T}$} \\
			0
			& \text{for $t \geq \widebar{T}$}
		\end{cases}
	\end{equation*}
	satisfies $\int_0^\infty \chi < \infty$ (so that it is Lebesgue-integrable),
	for then letting $T \uparrow \widebar{T}$ above yields the desired result by the dominated convergence theorem
	since $t \mapsto r \int_0^t e^{-rs}\phi^0(s) \dd s$ is $G$-integrable
	and $x,x^\dag$ are bounded.

	To show that $\int_0^\infty \chi$ is finite, note that applying the \hyperref[lemma:int_parts]{IBP lemma} as above to
	$x$ and $\widetilde{x}^\dag \coloneqq x + k$ (where $k \in \R$ is a constant)
	yields
	\begin{equation*}
		\int_{[0,T]} \left(
		r \int_0^t e^{-rs}\phi^0(s) \dd s
		\right) G(\dd t)
		= r \int_0^T [G(T)-G(t)]
		e^{-rt}\phi^0(t)
		\dd t
	\end{equation*}
	for any $T \in \R_+$ with $G(T) < 1$.
	Letting $T \uparrow \widebar{T}$
	and using the dominated (monotone) convergence theorem on the left-hand (right-hand) side yields
	\begin{equation*}
		\int_{\left[0,\widebar{T}\right)} \left(
		r \int_0^t e^{-rs}\phi^0(s) \dd s
		\right) G(\dd t)
		= \int_0^{\widebar{T}} \chi ,
	\end{equation*}
	which since $G$ has no atom at $\widebar{T}$ implies that
	\begin{equation*}
		\E_G\left( r \int_0^\tau e^{-rt} \phi^0(t) \dd t \right)
		= \int_0^{\widebar{T}} \chi
		= \int_0^{\infty} \chi .
	\end{equation*}
	Thus $\int_0^\infty \chi < \infty$
	since $t \mapsto r \int_0^t e^{-rs}\phi^0(s) \dd s$ is $G$-integrable.
\end{proof}%
\renewcommand{\qedsymbol}{$\blacksquare$}

\begin{proof}[Proof of {\hyperref[claim:gateaux_algebra_c]{claim C}}]%
	\renewcommand{\qedsymbol}{$\square$}
	Define a measure $\nu_+$ on $\R_+$ by
	\begin{equation*}
		\nu_+(A)
		\coloneqq \int_A \max\left\{\phi^1,0\right\} \dd G
		\quad \text{for any measurable $A \subseteq \R_+$.}
	\end{equation*}
	Since $\phi^1$ is $G$-integrable, $\nu_+$ is absolutely continuous with respect to (the measure associated with) $G$, with Radon--Nikod\'{y}m derivative $\max\{\phi^1,0\}$.
	Furthermore, the function $L : \R_+ \to \R$ defined by
	\begin{equation*}
		L(t) \coloneqq e^{-rt}\left( X^\dag_t - X_t \right)
		\quad \text{for each $t \in \R_+$}
	\end{equation*}
	is $\nu_+$-integrable since $X^\dag$ and $X$ are bounded.
	Thus for any $T \in \R_+$, we have
	\begin{equation*}
		\int_{[0,T]} L \max\left\{\phi^1,0\right\} \dd G
		= \int_{[0,T]} L \dd \nu_+ .
	\end{equation*}
	Furthermore, $L$ satisfies $L(t) = L(0) + \int_0^t l$, where $\smash{l(t) \coloneqq - e^{-rt} \bigl( x^\dag_t - x_t \bigr)}$.
	Thus for every $T \in \R_+$, the \hyperref[lemma:int_parts]{IBP lemma} yields
	\begin{multline*}
		\int_{[0,T]} e^{-rt}\left(X^\dag_t -X_t\right) \nu_+(\dd t)
		\\
		= e^{-rT} \left( X^\dag_T - X_T \right) \nu_+([0,T])
		+ r \int_0^T \nu_+([0,t])
		e^{-rt} \left( x^\dag_t - x_t \right) \dd t .
	\end{multline*}
	Applying the same argument to
	$\nu_-(A) \coloneqq - \int_A \min\left\{ \phi^1, 0 \right\} \dd G$
	and subtracting yields, for each $T \in \R_+$,
	\begin{multline*}
		\int_{[0,T]} e^{-rt}\left( X^\dag_t - X_t \right)
		\phi^1(t) G(\dd t)
		\\
		= e^{-rT} \left( X^\dag_T - X_T \right) \int_{[0,T]}
		\phi^1 \dd G
		+ r \int_0^T
		\left( \int_{[0,t]} \phi^1 \dd G \right)
		e^{-rt} \left( x^\dag_t - x_t \right) \dd t .
	\end{multline*}
	Letting $T \to \infty$ yields the desired equation
	by boundedness of $x,x^\dag,X,X^\dag$
	and dominated convergence.
\end{proof}%
\renewcommand{\qedsymbol}{$\blacksquare$}

With all three claims established,
the proof is complete.
\qed

\section*{\Cref*{suppl:construction:pf_lemma_approx_F} (p. \pageref*{suppl:construction:pf_lemma_approx_F})}
\label{omit:construction:pf_lemma_approx_F}
\addcontentsline{toc}{section}{Supplemental appendix \ref*{suppl:construction:pf_lemma_approx_F} (p. \pageref*{suppl:construction:pf_lemma_approx_F})}

Here's an example of a sequence of technologies
satisfying properties \ref*{bullet:thmC_pf_general:a}--\ref*{bullet:thmC_pf_general:d}.
Choose an $N \in \N$ such that $1/N < (u^0-u^1)/3$.
Given $n \geq N$, let $I_n \coloneqq [ 1/n , u^0 - 2/n ]$.
Fix $\gamma^n \in ( 0, 1/u^0n )$ and $\delta^n \in ( 0, 1/n )$.
For $j \in \{0,1\}$, define $f^j_n : I_n \rightarrow \R$ by
$\smash{f^j_n(u)
\coloneqq \frac{1}{\delta^n}\int_u^{u+\delta^n} F^{j+}
- \gamma^n u}$.
Given $\zeta^n \in (0,2/n)$, define $F^0_n,F^1_n : I_n \rightarrow \R$ by
\begin{align*}
	F^0_n(u)
	&\coloneqq F^0\left( u^\star + \frac{1}{n} \right)
	+ \int_{u^\star+1/n}^u f^0_n
	\\
	\text{and} \quad
	F^1_n(u)
	&\coloneqq F^1\left( u^\star + \frac{1}{n} \right)
	+ \int_{u^\star+1/n}^u f^1_n + \zeta^n .
\end{align*}
$F^0_n,F^1_n$ are strictly concave and differentiable, with (continuous) derivatives $f^0_n,f^1_n$ such that
\begin{equation*}
	F^{j-}(u+1/n) - 1
	\leq f^j_n(u)
	\leq F^{j+}(u)
	\quad \text{for $j \in \{0,1\}$ and $u \in I_n$.}
	\label{eq:derivative_bounds}
	\tag{$\natural$}
\end{equation*}
%
%
As $n \to \infty$,
$f^j_n \to F^{j+}$ pointwise for both $j \in \{0,1\}$ since $F^{j+}$ is right-continuous,
so that $F^0_n \to F^0$ and $F^1_n - \zeta^n \to F^1$ pointwise by the bounded convergence theorem.
In fact, the convergence is uniform \parencite[][Theorem 10.8, p. 90]{Rockafellar1970}.
Thus for any $\eps^n \in (0,1/n)$,
provided $\delta^n$ and $\gamma^n$ are sufficiently small,
we have that $F^0_n - F^0$ and $F^1_n - \zeta^n - F^1$ are bounded by $\eps^n$,
so that choosing $\zeta^n > 2 \eps^n$ ensures $F^1_n > F^0_n$ (as $F^1 \geq F^0$).
%

Since $F^1_n(1/n) > F^0_n(1/n)$, we may extend $F^0_n,F^1_n$ to $[0,u^0-2/n]$ while
preserving strict concavity, (continuous) differentiability and $F^1_n \geq F^0_n$
and ensuring that $F^1_n-F^0_n$ is strictly increasing on a neighbourhood of $0$
(which guarantees that $u^\star_n > 0$)
and that $F^{0\prime}_n,F^{1\prime}_n$ are bounded above on $(0,1/n]$ by $\max_{j \in \{0,1\}} F^{j+}(1/n) + 1$.
(The latter is possible because $F^{j\prime}_n(1/n) \leq F^{j+}(1/n)$ by \eqref{eq:derivative_bounds}.)
Further extend $F^1_n$ to all of $[0,\infty)$ by making it
differentiable and strictly concave on $\left[u^0-2/n,\infty\right)$, with derivative bounded below by $F^{1\prime}_n(u^0-2/n)-1/n$.
Since $F^1_n(u^0-2/n) > F^0_n(u^0-2/n)$, we may also extend $F^0_n$ to $[0,\infty)$ while
preserving strict concavity, (continuous) differentiability and $F^1_n \geq F^0_n$
and ensuring that $u^0_n \leq u^0$
and that $F^{0\prime}_n$ is bounded below by $F^{1-}(u^0-1/n)- 2$.
(The latter is possible since
$F^{1\prime}_n$ is bounded below by $F^{1\prime}_n(u^0-2/n)-1/n \geq F^{1-}(u^0-1/n)- 2$
by \eqref{eq:derivative_bounds}.)
If $u^\star_n$ is not a strict local maximum of $F^1_n-F^0_n$,
then we may make it so by decreasing $F^1_n$ pointwise on $[0,u^\star_n)$ by a small amount.

\begin{claim}
	\label{claim:monster}
	The sequence $(F^0_n,F^1_n)_{n=N}^\infty$ satisfies the model assumptions and properties \ref*{bullet:thmC_pf_general:a}--\ref*{bullet:thmC_pf_general:d}
	provided $\eps^n \in (0,1/n)$ is chosen sufficiently small.
\end{claim}

\begin{proof}
	It is clear that $(F^0_n,F^1_n)_{n=N}^\infty$ satisfies property \ref*{bullet:thmC_pf_general:a},
	and that for each $n \in \N$, $F^0_n,F^1_n$ satisfy the all of the model assumptions except for $u^1_n < u^0_n$ (the conflict of interest).

	For \ref*{bullet:thmC_pf_general:b} and the conflict of interest, since and $u^0$, $u^1$ and $u^\star$ are strict local maxima of (respectively) $F^0$, $F^1$ and $F^1-F^0$, we have
	\begin{equation*}
		u^0 - \frac{1}{n} \leq u^0_n \leq u^0 ,
		\quad
		\abs*{u^1_n-u^1} \leq \frac{1}{n}
		\quad \text{and} \quad
		\abs*{u^\star_n-u^\star} \leq \frac{1}{n}
	\end{equation*}
	for $\eps^n > 0$ sufficiently small, as $F^0_n - F^0$ and $F^1_n - (F^1+\zeta^n)$ are bounded by $\eps^n$ on $I_n$.
	In particular, $u^1_n < u^0_n$ (the conflict of interest) since $n \geq N$.
	Letting $n \to \infty$ yields \ref*{bullet:thmC_pf_general:b}.

	For the first part of \ref*{bullet:thmC_pf_general:c},
	fix a $u \in (0,u^0]$ at which $F^{1-}$ is finite.
	If $u=u^0$,
	then for each $n \geq N$,
	$F^{0\prime}_n, F^{1\prime}_n$ are bounded below by $F^{1-}(u^0-1/n) - 2$
	(by construction for $F^{0\prime}_n$, and by (\ref{eq:derivative_bounds}) for $F^{1\prime}_n$).
	Suppose instead that $u<u^0$.
	Choose a $u' \in (u,u^0)$, and let $N' \geq N$
	satisfy $1/N' < u < u^0 - 2/N'$ and $1/N' < u'-u$.
	Then $u$ belongs to $I_n$ for every $n \geq N'$,
	so that $F^{j\prime}_n(u) \geq F^{j-}(u+1/n) - 1 \geq F^{j-}(u') - 1$ by (\ref{eq:derivative_bounds}) for both $j \in \{0,1\}$.
	The set
	\begin{equation*}
		\Union_{j \in \{0,1\}}
		\left\{ F^{j\prime}_n(u) : n \in \{1,\dots,N'\} \right\}
	\end{equation*}
	is bounded since it is finite.
	Thus $(F^{0\prime}_n)_{n \in \N}$ and $(F^{1\prime}_n)_{n \in \N}$ are uniformly bounded below on $[0,u]$.

	The argument for the second part of \ref*{bullet:thmC_pf_general:c} is analogous:
	fix a $u \in [0,u^0)$ at which $F^{0+},F^{1+}$ are finite.
	If $u=0$,
	then $F^{0\prime}_n,F^{1\prime}_n$ are bounded above by $\max_{j \in \{0,1\}} F^{0+}_j(1/n)+1$ for all $n \geq N$
	by construction.
	If instead $u>0$,
	then $F^{j\prime}_n(u) \leq F^{j+}(u)$ for all $n \geq N$ such that $1/n < u < u^0 - 2/n$ by (\ref{eq:derivative_bounds}),
	and the set
	\begin{equation*}
		\Union_{j \in \{0,1\}}
		\left\{ F^{j\prime}_n(u) :
		n \in \left\{1, \dots,
		\floor*{1/u} \join
		\floor*{ \left. 2 \middle/ \left(u^0-u\right) \right. }
		\right\} \right\}
	\end{equation*}
	is finite and thus bounded;
	so $(F^{0\prime}_n)_{n \in \N}$ and $(F^{1\prime}_n)_{n \in \N}$ are uniformly bounded above on $\left[u,u^0\right]$.

	For \ref*{bullet:thmC_pf_general:d}, fix $j \in \{0,1\}$ and a convergent sequence $(u_n)_{n=N}^\infty$ with limit $u \in [0,\infty)$ such that $0 < u_n \leq u^0_n$ for each $n \geq N$ and $\lim_{n \to \infty} F^{j\prime}_n(u_n)$ exists.
	Since $u^0_n \leq u^0$ for each $n \geq N$, we have $u \leq u^0$.
	From above,
	\begin{equation*}
		F^{j\prime}_n(u_n)
		= f^j_n(u_n)
		\;
		\begin{cases}
			\leq F^{j+}(u_n)
			& \text{if $u_n \geq 1/n$} \\
			\geq F^{j-}(u_n+1/n) - \gamma^n u_n
			& \text{if $u_n \leq u^0-2/n$}.
		\end{cases}
		\label{eq:monster_fn_dag}
		\tag{$\dag$}
	\end{equation*}
	Letting $n \to \infty$ yields
	\begin{equation*}
		\lim_{n \to \infty} F^{j\prime}_n(u_n)
		\;
		\begin{cases}
			\leq F^{j-}(u)
			& \text{if $u>0$} \\
			\geq F^{j+}(u)
			& \text{if $u<u^0$.}
		\end{cases}
	\end{equation*}
	If $u \in (0,u^0)$, then both cases are satisfied,
	so
	\begin{equation*}
		F^{j+}(u) \leq \lim_{n \to \infty} F^{j\prime}_n(u_n) \leq F^{j+}(u) .
	\end{equation*}
	If $u=0$, then the second case applies, so that
	$F^{j+}(0) \leq \lim_{n \to \infty} F^{j\prime}_n(u_n)$,
	which suffices.
	If $u=u^0$, then the first case holds, so it remains only to show that $F^{j+}( u^0 ) \leq \lim_{n \to \infty} F^{j\prime}_n(u_n)$.
	For $j=0$, we have
	$F^{0+}(u^0) \leq 0 \leq \lim_{n \to \infty} F^{0\prime}_n(u_n)$
	since $u_n \leq u^0_n$ for each $n \geq N$.
	For $j=1$, \eqref{eq:monster_fn_dag} implies
	\begin{multline*}
		F^{1-}\left( \min\left\{ u_n + 1/n, u^0 - 1/n \right\} \right)
		- \gamma^n u_n - 1/n
		\quad
		\\
		\leq F^{1\prime}_n\left( \min\left\{ u_n, u^0 - 2/n \right\} \right) - 1/n
		\leq F^{1\prime}_n(u_n)
		\quad \text{for each $n \geq N$,}
	\end{multline*}
	where the second inequality holds since $F^{1\prime}_n$ is bounded below by $F^{1\prime}_n(u^0-2/n)-1/n$. Letting $n \to \infty$ yields
	$F^{1+}( u^0 )
	\leq \lim_{n \to \infty} F^{1\prime}_n(u_n)$.
\end{proof}

\section*{\Cref*{suppl:mcs:pf_lemma_euler_optimal} (p. \pageref*{suppl:mcs:pf_lemma_euler_optimal})}
\label{omit:mcs:pf_lemma_euler_optimal}
\addcontentsline{toc}{section}{Supplemental appendix \ref*{suppl:mcs:pf_lemma_euler_optimal} (p. \pageref*{suppl:mcs:pf_lemma_euler_optimal})}

\begin{proof}[Proof of \Cref*{observation:piG_str_conc}]
	It suffices to show that $\pi_G$ is strictly concave.
	For distinct $x,x^\dag \in \mathcal{X}$ (i.e. $x \neq x^\dag$ on a non-null set) and any $\lambda \in (0,1)$, we have
	\begin{multline*}
		\pi_G\left( \lambda x + (1-\lambda) x^\dag \right)
		\\
		\begin{aligned}
			&= \E_G\left( r \int_0^\tau e^{-rt}
			F^0\left( \lambda x_t + (1-\lambda) x^\dag_t \right) \dd t
			+ e^{-r\tau} F^1\left( \lambda X_\tau + (1-\lambda) X^\dag_\tau \right)
			\right)
			\\
			&> \E_G\biggl( r \int_0^\tau e^{-rt}
			\left[ \lambda F^0(x_t)
			+ (1-\lambda) F^0\left( x^\dag_t \right) \right] \dd t
			\\
			&\qquad\qquad\quad + e^{-r\tau} \left[ \lambda F^1(X_\tau)
			+ (1-\lambda) F^1\left( X^\dag_\tau \right) \right]
			\biggr)
			\\
			&= \lambda \pi_G(x) + (1-\lambda) \pi_G\left( x^\dag \right) ,
		\end{aligned}
	\end{multline*}
	where the inequality holds since $F^1$ is concave,
	$F^0$ is strictly concave, $G$ has unbounded support and $x \neq x^\dag$ on a non-null set.
\end{proof}



\printbibliography[heading=bibintoc]


\end{document}